\documentclass[a4paper,USenglish]{lipics-v2016}
 
\usepackage{microtype}
\usepackage{etex}

\usepackage{float}
\usepackage{wrapfig}
\usepackage{graphicx}
\usepackage{amssymb,amsmath,amsthm}
\usepackage[mathscr]{eucal}
\usepackage{verbatim}

\usepackage{paralist}
\usepackage{mathtools}
\usepackage{amsfonts}
\usepackage{color}
\usepackage{tikz}
\usepackage{placeins}
\usetikzlibrary{arrows,fit,shapes,automata}
\usetikzlibrary{positioning,fit,calc,shapes}
\usetikzlibrary{decorations.fractals}
\usetikzlibrary{decorations.markings}
\usepackage{xspace} 

\usepackage[usenames,dvipsnames]{pstricks}
\usepackage{epsfig}
\usepackage{pst-grad} 
\usepackage{pst-plot} 
\usepackage{pstricks-add}

\theoremstyle{plain}
\newtheorem{proposition}[theorem]{Proposition}

\title{The Robot Routing Problem for Collecting Aggregate Stochastic Rewards
}
\titlerunning{The Robot Routing Problem for Collecting Aggregate Stochastic Rewards} 

\author[1]{Rayna Dimitrova, Ivan Gavran, Rupak Majumdar,\\ Vinayak S.~Prabhu, and Sadegh Esmaeil Zadeh Soudjani}
\affil[1]{
Max Planck Institute for Software Systems, Kaiserslautern, Germany\\
  \texttt{\{Rayna,Gavran,Rupak,Vinayak,Sadegh\}@mpi-sws.org}}
\authorrunning{R. Dimitrova, I. Gavran, R. Majumdar, V.\,S. Prabhu, and S.\,E.\,Z. Soudjani} 

\Copyright{Rayna Dimitrova, Ivan Gavran, Rupak Majumdar, Vinayak S.~Prabhu, and Sadegh  Soudjani}


\newcommand{\mc}{\mathcal}

\newcommand{\nat}{{\mathbb{N}}}

\newcommand{\set}[1]{\{ #1 \}}                  
\newcommand{\ie}{\emph{i.e.}\xspace}
\newcommand{\eg}{\emph{e.g.}\xspace}

\DeclareMathOperator{\ssum}{\mathsf{sum}}
\newcommand{\rsum}{r_{\ssum}}

\newcommand{\csum}{c_{\ssum}}

\DeclareMathOperator{\sav}{\mathsf{av}}
\newcommand{\rav}{r_{\sav}}
\newcommand{\cav}{c_{\sav}}

\newcommand{\Rsum}{R_{\ssum}}
\newcommand{\Rav}{R_{\sav}}

\newcommand{\Cav}{C_{\sav}}
\newcommand{\Csum}{C_{\ssum}}

\newcommand{\ocsum}{\bar{c}_{\mathsf{sum}}}
\newcommand{\ucsum}{\underline{c}_{\,\mathsf{sum}}}

\DeclareMathOperator{\mysum}{\mathsf{Sum}}
\DeclareMathOperator{\acc}{\mathsf{acc}}
\DeclareMathOperator{\Eacc}{\mathsf{Eacc}}

\DeclareMathOperator{\outcome}{\mathsf{outcome}}

\DeclarePairedDelimiter\abs{\lvert}{\rvert}%
\DeclareMathOperator{\Last}{\mathsf{Last}}

\newcommand{\Inf}{{\mathsf{Inf}}}

\begin{document}

\maketitle

\begin{abstract}
We propose a new model for formalizing reward collection problems on graphs
with dynamically generated rewards which may appear and disappear
based on a stochastic model.
The \emph{robot routing problem} is modeled as a graph whose nodes are stochastic
processes generating  potential rewards over discrete time.
The rewards are generated according to the stochastic process,
but at each step, an existing reward disappears with a given probability.
The edges in the graph encode the (unit-distance) paths between the rewards' locations.
On visiting a node, the robot collects the accumulated reward at the node at that time,
but traveling between the nodes takes time.
The optimization question asks to compute an optimal (or $\epsilon$-optimal) path  that maximizes the expected collected rewards.
\looseness=-1

We consider the finite and infinite-horizon robot routing problems.
For finite-horizon, the goal is to maximize the total expected reward, while for infinite
horizon we consider limit-average objectives. We study the computational
and strategy complexity of these problems, establish NP-lower bounds and
show that optimal strategies require memory in general. We also provide an
algorithm for computing $\epsilon$-optimal infinite paths for arbitrary
$\epsilon > 0$.

 \end{abstract}
 
 \section{Introduction}\label{sec:intro}
 
 Reward collecting problems on metric spaces are at the core of many applications,
 and studied classically in combinatorial optimization under many well-known monikers:
 the traveling salesman problem, the knapsack problem, the vehicle routing problem, 
 the orienteering problem, and so on.
 Typically, these problems model the metric space as a discrete graph whose nodes
 or edges constitute rewards, either deterministic or stochastic, and ask
 how to traverse the graph to maximize the collected rewards.
 In most versions of the problem, rewards are either fixed or cumulative.
 In particular, once a reward appears, it stays there until collection. 
 However, in many applications, existing rewards may disappear (e.g., a customer
 changing her mind) or have more ``value'' if they are collected fast.
 
 We introduce the \emph{Robot Routing problem}, which combines the spatial aspects of traveling
 salesman and other reward collecting problems on graphs with stochastic reward generation
 and with the possibility that uncollected rewards may disappear at each stage.
 The robot routing problem consists of a finite graph and a reward process for each node of the graph.
 The reward process models dynamic requests which appear and disappear.
 At each (discrete) time point, a new reward is generated
 for the node according to a stochastic process with expectation $\lambda$. 
 However, at each point, a previously generated reward disappears with a fixed probability $\delta$.
 When the node is visited, the entire reward is collected.
 The optimization problem for robot routing asks, given a graph and a reward process,
 what is the optimal (or $\epsilon$-optimal)  path a robot should traverse in this graph
 to maximize the expected reward?
 
 As an illustrating example for our setting, consider a vendor planning her path through a city. At each street corner, and at each
 time step, a new customer arrives with expectation $\lambda$, and an existing customer leaves with
 probability $\delta$. When the vendor arrives at the corner, she serves all the existing requests at once.
 We ignore other possible real-world features and behaviors  e.g., customers leaving queues if the queue length is long.
 How should the vendor plan her path?
 Similar problems can be formulated for traffic pooling \cite{Wongpiromsarn13}, for robot control \cite{HSCC2017},
 for patrolling \cite{HoshinoU16}, and many other scenarios.
 
 Despite the usefulness of robot routing in many scenarios involving dynamic appearance and disappearance of rewards,
 algorithms for its solution have not, to the best of our knowledge, been studied before.
 In this paper, we study two optimization problems:
 the \emph{value computation problem}, that asks for the maximal expected reward over a \emph{finite} or \emph{infinite}
 horizon, and the \emph{path computation problem}, that  asks for a path realizing the optimal  (or $\epsilon$-optimal) reward.
 The key observation to solving these problems
 is that the reward collection can be formulated as discounted sum problems over an extended
 graph, using the correspondence between stopping processes and discounted sum games. 
 
 For finite horizon robot routing we show that the value \emph{decision} problem (deciding if the maximal expected reward 
 is at least a certain amount) is NP-complete when the horizon bound is given in unary, 
 and the value and optimal path can be computed in exponential time using dynamic programming.
 
 For the infinite horizon problem, where the accumulated reward is defined as the long run average,
 we show that the value decision problem is NP-hard if the probability of a reward disappearing is
 more than a threshold dependent on the number of nodes.
 We show that computing the optimal long run average reward can be reduced 
 to a 1-player mean-payoff game on an \emph{infinite graph}. 
 By solving the mean payoff game on a finite truncation of this graph, we can approximate the 
 solution up to an arbitrary precision. 
 This gives us an algorithm that, for any given $\epsilon$,
 computes an $\epsilon$-optimal path
 in time exponential in the size of the original graph and logarithmic in $1/\epsilon$.
 Unlike finite mean-payoff 2-player games, 
 strategies which generate 
 optimal paths for robot routing even in the 1-player setting can require memory. 
 For the \emph{non-discounted} infinite horizon problem (that is, when rewards do not disappear)
 we show that the optimal path and value problems are solvable in polynomial time.\looseness=-1

\vspace*{-4mm}
\subparagraph*{Related work}
 The robot routing problem is similar in nature to a number of other problems studied in robot 
 navigation, vehicle routing, patrolling, and queueing network control, 
 but to the best of our knowledge has not been studied so far. 
 
 There exists a plethora of versions of the famous traveling salesman problem (TSP) which 
 explore the trade-off between the cost of the constructed path and its reward. 
 Notable examples include the orienteering problem~\cite{VansteenwegenSO11}, in which the number of locations visited in a 
 limited amount of time is to be maximized, vehicle routing with time-windows~\cite{Kolen87} and deadlines-TSP~\cite{BansalBCM04}, 
 which impose restrictions or deadlines on when locations should be visited, as well as discounted-reward-TSP~\cite{BlumCKLMM07} 
 in which soft deadlines are implemented by means of discounting. 
 Unlike in our setting, in all these problems, rewards are static, and there is no generation and accumulation of rewards, which is a key feature of our model.
 
 In the dynamic version of vehicle routing~\cite{BulloFPSS11} and the dynamic traveling repairman problem~\cite{BertsimasR91}, 
 tasks are dynamically introduced and the objective is to minimize the expected task waiting time. 
 In contrast, we focus on limit-average objectives, which are a classical way to combine rewards over infinite system runs.
 Patrolling~\cite{BrazdilHKRA15} is another graph optimization problem, motivated by operational security. 
 The typical goal in patrolling is to synthesize a strategy for a defender against 
 a single attack at an undetermined time and location, and is thus incomparable to ours. 
 A single-robot multiple-intruders patrolling setting that is close to ours is described in~\cite{HoshinoU16}, but there again the objective is to merely detect whether there is a visitor/intruder at a given room. Thus, the patrolling environment in~\cite{HoshinoU16} is described by the probability of detecting a visitor for each location. On the contrary, our model can capture \emph{counting patrolling problems}, where the robot is required not only to detect the presence of visitors but to register/count as many of them as possible.
 Another related problem is the information gathering problem~\cite{Stranders13}. The key difference between the information gathering setting and ours is that~\cite{Stranders13} assumes that making an observation earlier has bigger value than if a lot of observations have already been made. This restriction on the reward function is \emph{not} present in our model, since the reward value collected when visiting node $v$ at time $t$ (making observation $(v,t)$, in their terms) only depends on the last time when $v$ was previously visited, and not on the rest of the path (the other observations made, in their terms).
 
 Average-energy games~\cite{ChatterjeeP15, Bouyer2017} are a class of games on finite graphs in which the limit-average objective is defined by a double summation.
 The setting discussed in~\cite{ChatterjeeP15, Bouyer2017} considers static edge weights and no discounting.
 Moreover, the inner sum in an average-energy objective is over the whole prefix so far, while in our setting the inner sum spans from the last to the current visit of the current node, which is a crucial difference between these two settings.

 Finally, there is a rich body of work on multi-robot routing~\cite{WeiHJ16a,AmadorOZ14,MelvinKKTO07,EkiciKK09,EkiciR13} which is closely related to our setting. 
 However, the approaches developed there are limited to static tasks with fixed or linearly decreasing rewards. 
 The main focus in the multi-robot setting is the task allocation and coordination between robots, which is a dimension orthogonal to the aggregate reward collection problem which we study.
 
 Markov decision processes (MDP) \cite{Puterman94} seem superficially close to our model. 
 In an MDP, the rewards are determined statically as a function of the state and
 action. 
 In contrast, the dynamic generation and accumulation of rewards in our model, especially the
 individual discounting of each generated reward, leads 
 to algorithmic differences: 
 for example, while MDPs admit memoryless strategies for long run average
 objectives, strategies require memory in our setting and there is no obvious
 reduction to, e.g., an exponentially larger, MDP.
 
 We employed the reward structure of this article in \cite{HSCC2017} with the goal of synthesizing controllers for reward collecting Markov processes in continuous space.
 The work \cite{HSCC2017} is mainly focused on addressing the continuous dynamics of the underlying Markov process where the authors use abstraction techniques \cite{SA13} to provide approximately optimal controllers with formal guarantees on the performance while maintaining the probabilistic nature of the process. In contrast, we tackle the challenges of this problem with having a deterministic graph as the underlying dynamical model of the robot and study the computational complexity of the proposed algorithms thoroughly.

\vspace*{-2mm}
\subparagraph*{Contributions}
 We define a novel optimization problem for formalizing and solving reward collection in a metric
 space where stochastic rewards appear as well as disappear over time.
 \begin{itemize}
 	\item We consider reward-collection problems in a novel model with \emph{dynamic generation} and \emph{accumulation} of rewards, 
 	where each reward \emph{can disappear with a given probability}.
 	\item We study the value decision problem,  the value computation problem, and the path computation problem over a finite horizon. We show that the value decision problem is NP-complete when the horizon is given in unary. We describe a dynamic programming approach for computing the optimal value and an optimal path in exponential time.
 	\item We study the value decision problem,  the value computation problem, and the path computation problem over an infinite horizon. We show that for sufficiently large values of the disappearing factor $\delta$ the value decision problem is NP-hard. 
 	We provide an algorithm which for any given $\epsilon > 0$, computes an $\epsilon$-optimal path in time exponential in the size of the original graph and 
 	logarithmic in $1/\epsilon$. We demonstrate that strategies (in the 1-player robot routing games) which generate infinite-horizon optimal paths can require memory. 
 	
 \end{itemize}

 \section{Problem Formulation}\label{sec:problem}

\vspace*{-2mm}
\subparagraph*{Preliminaries and notation}
 A finite directed graph $G  = (V,E)$ consists of a 
 finite set of nodes  $V$ and a set of edges $E \subseteq V \times V$.
 A path $\pi = v_0,v_1,\ldots$ in $G$ is a finite or infinite sequence of nodes in $G$, 
 such that $(v_i,v_{i+1}) \in E$ for each $i  \geq 0$.
 We denote with $|\pi| = N$ the length (number of edges) of a finite path $\pi = v_0,v_1,\ldots,v_N$ and 
 write $\pi[i]= v_i$ and $\pi[0\ldots N] = v_0,v_1,\ldots v_N$.
 For an infinite path $\pi$, we define $|\pi| = \infty.$
 We also denote the cardinality of a finite set $U$ by $|U|$. 
 We denote by $\mathbb N\!=\!\{0,1,..\}$ and 
 $\mathbb Z_+ \!= \!\{1,2,..\}$ the sets of
 non-negative and positive integers respectively.
 We define $\mathbb Z[n,m] = \{n,n+1,\ldots,m\}$ for any $n,m\in\!\mathbb N,\,n\le m$.
 We denote with $\mathbb I(\cdot)$ the indicator function which takes a 
 Boolean-valued expression as its argument and returns $1$ if 
 this expression evaluates to true and 0 otherwise.
 
 \vspace*{-4mm}
 \subparagraph*{Problem setting}
 Fix a graph $G = (V,E)$.
 We consider a discrete-time setting where at each time step  $t\in \mathbb{N}$,
 at each node $v\in V$ a reward process generates rewards according to some probability distribution. 
 Once generated, each reward at a node decays according to a decaying function.
 A \emph{reward-collecting} robot starts out at some node $v_0\in V$
 at time $t=0$, and traverses one edge in $E$ at each time step.
 Every time the robot arrives at a node $v\in V$, it  collects the reward accumulated at $v$ since
 the last visit to $v$.
 Our goal is to compute the maximum expected reward that the robot can possibly collect, and to 
 construct an optimal path for the robot in the graph, i.e., a path whose expected total reward is maximal.
 
 To formalize reward accumulation, we define a function $\Last_{\pi}$ which (for path $\pi$)  maps
 an index $t \leq |\pi|$ and a node $v \in V$ to 
 the length of the path starting at the previous occurrence of $v$ in $\pi$ till position $t$;
 and to $t+1$ if $v$ does not occur in $\pi$ before time $t$:
 \begin{equation*}
 \Last_\pi(t,v) := \min\left(t+1, \ \left\{t-j \in \mathbb{N} \mid j < t,  \pi[j] = v\right\}\right).
 \end{equation*}
 
 \vspace*{-4mm}
 \subparagraph*{Reward functions}
 Let $\xi : \Omega\times V\to \mathbb R$ be a set of random variables defined on a sample space $\Omega$ and indexed by the set of nodes $V$. Then $\xi(\cdot,v)$, $v\in V,$ is a measurable function from $\Omega$ to $\mathbb R$ that generates a random reward at node $v$ at any time step.
 Let $\pi$ be the path in $G$ traversed by the robot.
 At time $t$, the position of the robot is the node $\pi[t]$, and the robot 
 collects the uncollected decayed reward generated at node $\pi[t]$ (since its last visit to $\pi[t]$)
 up till and including time $t$. Then, the robot traverses the edge $(\pi[t], \pi[t+1])$, and at time $t+1$ 
 it collects the rewards at node $\pi[t+1]$.
 
 The uncollected reward at time $t$ at a node $v$ given a path $\pi$ traversed by the robot
 is defined by the random variable
 \begin{equation*}
 \acc_{\pi}(t, v)\  := \hspace{-0.1in}\sum_{j = 0}^{\Last_\pi(t,v)-1}\hspace{-0.1in}\gamma(v)^{j} \xi(w(t-j),v),\quad w(\cdot)\in\Omega.
 \end{equation*} 
 
 The value $\gamma(v)$ in the above definition is a \emph{discounting factor} that models the probability that a reward at node $v$ survives for one more round, that is, the probability that a given reward instance at node $v$ disappears at any step is $1-\gamma(v)$.
 
 Note that the previous time a reward was collected at node $v$ was at time  $t- \Last_\pi(v, t)$, 
 the time node $v$ was last visited before $t$.
 Thus $\acc_{\pi}(t, v)$ corresponds to the rewards generated at node $v$ at times
 $t,\, t-1, \dots,\,  t- \Last_\pi(t,v)+1$, which have decayed by factors of
 $\gamma(v)^0, \gamma(v)^1, \dots, \gamma(v)^{\Last_\pi(t,v)-1},$ respectively.
 When traversing a path $\pi$, the robot collects the accumulated reward $\acc_{\pi}(t, \pi[t])$ at time $t$ at node $\pi[t]$.
 
 We define the \emph{expected finite $N$-horizon sum reward} for a path $\pi$ as:
 \begin{equation*}
 \rsum^{(N)}(\pi) \ :=\ \mathbb E\left[\sum_{t=0}^N \acc_{\pi}(t, \pi[t])\right].
 \end{equation*}
 
 Let $\lambda : V   \to \mathbb{R}_{\geq 0}$ be a function that maps each node $v \in V$  to the \emph{expected value of the reward generated at node $v$} for each time step, $\lambda(v) = \mathbb E\left[\xi(\cdot,v)\right]$.
 We assume that the rewards generated at each node are independent of the agent's move. Thus, the function $\lambda$ will be sufficient for our study, since we have
 \begin{align}
 &\rsum^{(N)}(\pi) \ =\ \sum_{t=0}^N \Eacc_{\pi}(t, \pi[t])  \text{, where }
 \Eacc_{\pi}(t, v):=\hspace{-0.1in} \sum_{j = 0}^{\Last_\pi(t,v)-1}\hspace{-0.1in}\gamma(v)^{j} \lambda(v).
 \label{eq:rsum-orig}
 \end{align}
 
 For an infinite path $\pi$, the \emph{limit-average} expected reward is defined as
 \begin{equation}
 \label{eq:rav-orig}
 \rav(\pi) \ =\  \liminf_{N \to \infty}\frac{\rsum^{(N)}\left( \pi\right)}{N+1}.
 \end{equation}
 The finite and infinite-horizon \emph{reward values} for a node $v$ are defined as the best rewards over all paths originating in $v$: $\Rsum^{(N)}(v) = \sup_{\pi}\left\{\rsum^{(N)}(\pi)\,|\,\pi[0]=v,|\pi| = N\right\}$ and
 $\Rav(v) =  \sup_{\pi}\left\{\rav(\pi)\,|\,\pi[0]=v,|\pi| = \infty\right\}$, respectively.
 The choice of limit-average in \eqref{eq:rav-orig} is due to the unbounded sum reward $\rsum^{(N)}(\pi)$ when $N$ goes to infinity.
 For a given path $\pi$, the sequence $\rsum^{(N)}\left( \pi\right)/(N\!+\!1)$  in \eqref{eq:rav-orig} may not converge.
 Thus we opt for the worst case limiting behavior of the sequence. 
 Alternatively, one may select the best case limiting behavior 
 $\limsup_{N \to \infty}$ in \eqref{eq:rav-orig} with no substantial change in the results of this paper.
 
 \vspace*{-4mm}
 \subparagraph*{Node-invariant functions $\lambda$ and $\gamma$ and definition of cost functions}
 In the case when the functions $\lambda$ and $\gamma$ are constant, we write  $\lambda$ and $\gamma$ for the respective constants.
 In this case, the expressions for $\rsum^{(N)}(\pi)$ and $\rav(\pi)$ can be simplified using the
 identity $1+\gamma + \gamma^2 + \dots + \gamma^{q-1} = \frac{1-\gamma^q}{1-\gamma}$ for $\gamma < 1$.
 Then we have
 \begin{equation}
 \label{eq:rsumNew}
 \begin{split}
 \rsum^{(N)}(\pi)& = \sum_{t=0}^N\sum_{j = 0}^{\Last_\pi(\pi[t],t)-1}\gamma^{j} \lambda
 =  \lambda\cdot \sum_{t=0}^N 
 \left(1 + \gamma +\ldots+\gamma^{\Last_\pi(t,\pi[t])-1}\right)\\
 &
 =   \lambda\cdot \sum_{t=0}^N  \frac{1 - \gamma^{\Last_\pi(t,\pi[t])}}{1-\gamma}
 = \frac{(N+1)\lambda}{1-\gamma}
 -  \frac{\lambda}{1-\gamma}\sum_{t=0}^N \gamma^{\Last_\pi(t,\pi[t])}.
 \end{split}
 \end{equation}
 The expression $\rav(\pi)$ can be simplified as:
 \begin{equation}
 \label{eq:ravNew}
 \rav(\pi) = \liminf_{N \to \infty}\frac{1}{N+1}\rsum^{(N)}(\pi)
 = 
 \frac{\lambda}{1-\gamma} - \frac{\lambda}{1-\gamma}\limsup_{N \to \infty}\frac{1}{N+1}
 \sum_{t=0}^N \gamma^{\Last_\pi(t,\pi[t])}.
 \end{equation}
 For the special case $\gamma=1$ (i.e., when the rewards are not discounted), the expression for the finite-horizon reward is
 $ \rsum^{(N)}(\pi) = \lambda\sum_{t=0}^{N}\Last_\pi(t,\pi[t])$.
 
 We define \emph{cost functions} that map a path $\pi$ to a real valued finite- or infinite-horizon cost:\looseness=-1
 \begin{equation}
 \label{eq:cost_functions}
 \csum^{(N)}(\pi) := \sum_{t=0}^{N} \gamma^{\Last_\pi(\pi[t],t)}\quad \text{ and }\quad
 \cav(\pi) := \limsup\limits_{N \to \infty}\frac{ \csum^{(N)}(\pi)}{N+1}.
 \end{equation}
 From Equations~\eqref{eq:rsumNew} and~\eqref{eq:ravNew}, the computation of optimal paths
 for the reward functions $\rsum^{(N)}$ and $\rav$ corresponds to computing
 paths that minimize the cost functions $\csum^{(N)}(\pi)$ and  $\cav(\pi)$, respectively.
 Analogously to $\Rsum^{(N)}(v)$ and $\Rav(v)$, the infimums of the cost functions in \eqref{eq:cost_functions} over paths are denoted by  $\Csum^{(N)}(v)$ and $\Cav(v)$ respectively.
 
 \begin{wrapfigure}{l!}{0.4\textwidth}
 	\centering
 	\begin{tikzpicture}[node distance=2 cm,auto,>=latex',line join=bevel,transform shape]
 	\node[circle,draw] at (0,0) (a) {$a$};
 	\node  [left of=a,circle,draw] (d) {$d$};
 	\node  [above right of =a,yshift=-.7cm,circle,draw] (b) {$b$};
 	\node  [below right of =a,yshift=.7cm,circle,draw] (c) {$c$};
 	\draw [->] (a) edge[bend left] (b);
 	\draw [->] (b) edge[bend left] (c);
 	\draw [->] (c) edge[bend left] (a);
 	\draw [->] (a) edge[bend right] (d);
 	\draw [->] (d) edge[bend right] (a);
 	\end{tikzpicture}
 	\caption{A graph $G_{\mathsf e} = (V_{\mathsf e},E_{\mathsf e})$ with two simple cycles  sharing a single node.}
 	\label{fig:graph1}
 \end{wrapfigure}
 
 \begin{example}
 	\label{ex:4node-reward}
 	Consider the graph $G_{\mathsf e} = (V_{\mathsf e},E_{\mathsf e})$ in Figure \ref{fig:graph1} with $V_{\mathsf e} = \{a,b,c,d\}$, which we will use as a running example throughout the paper. 
 	The functions $\lambda$ and $\gamma$ are constant.

 \end{example}

 Consider the finite path $\pi_1 = adabcad$. For the occurrences of node $a$ in $\pi_1$ we have 
 $\Last_{\pi_1}(0,a) = 1$, $\Last_{\pi_1}(2,a) = 2$, $\Last_{\pi_1}(5,a) = 3$, and similarly for the other nodes in $\pi_1$. The reward for $\pi_1$ as a function of $\lambda$ and $\gamma$ is 
 $\rsum^6(\pi_1) = \frac{7\lambda}{1-\gamma}
 -  \frac{\lambda}{1-\gamma}(\gamma + \gamma^2 + \gamma^2 + \gamma^4 + \gamma^5 + \gamma^3 + \gamma^5)$
 for $\gamma<1$ and $\rsum^6(\pi_1) = 22\lambda$ for $\gamma=1$.
 For the infinite path $\pi_2 = (abc)^\omega$ we have $\Last_{\pi_2}(0,a) = 1$, $\Last_{\pi_2}(1,b) = 2$, $\Last_{\pi_2}(2,c) = 3$ and 
 the value of $\Last$ is $3$ in all other cases. Thus we have 
 $\rav(\pi_2) = \frac{\lambda}{1-\gamma} - \frac{\lambda}{1-\gamma}\gamma^3$
 for $\gamma<1$ and $\rav(\pi_2) = 3\lambda$ for $\gamma=1$.
 Similarly, for $\pi_3 = (abcad)^\omega$ we have $\rav(\pi_3) = \frac{\lambda}{1-\gamma} - \frac{\lambda}{1-\gamma} \cdot \frac{(\gamma^2+\gamma^3+3\gamma^5)}{5}$
 for $\gamma<1$ and $\rav(\pi_3) = 4\lambda$ for $\gamma=1$.

 \subparagraph*{Problem statements}
 We investigate optimization and decision problems for finite and infinite-horizon robot routing.
 The \emph{value computation problems} ask for the computation of $\Rsum^{(N)}(v)$ and $\Rav(v)$. The corresponding decision problems asks to check if the respective one of these two quantities is greater than or equal to a given threshold $R\in\mathbb R$.\looseness=-1
 
 \begin{definition}[Value Decision Problems]
 	\label{def:ValueDecision}
 	Given 
 	a finite directed graph $G = (V,E)$, 
 	an expected reward function $\lambda:V\rightarrow\mathbb R_{\ge 0}$, 
 	a discounting function $\gamma:V\rightarrow (0,1]$,
 	an initial node $v_0 \in V$ and 
 	a threshold value $R \in \mathbb{R}$,
 	\begin{compactitem}
 		\item The \emph{finite horizon value decision problem} is to decide,
 		given $N$, if
 		$\Rsum^{(N)}(v_0) \geq R$.
 		\item The \emph{infinite horizon value decision problem} is to decide if
 		$\Rav(v_0) \geq R$.
 	\end{compactitem}
 \end{definition}

 For  a finite directed graph $G = (V,E)$,
 expected reward  and discounting functions $\lambda:V\rightarrow\mathbb R_{\ge 0}$ and
 $\gamma:V\rightarrow (0,1]$ and  $v_0 \in V$, a finite path $\pi$ is said to be an
 \emph{optimal path} for time-horizon $N$ if 
 (a)~$\pi[0] = v_0$ and $|\pi|=N$, and
 (b)~for every path $\pi'$ in $G$ with $\pi'[0] = v_0$ and $|\pi'|=N$ it holds that $\rsum^{(N)}(\pi) \geq \rsum^{(N)}(\pi')$.
 Similarly, an infinite path $\pi$ is said to be \emph{optimal for the infinite horizon} if $\pi[0] = v_0$ and for every infinite path $\pi'$ with $\pi'[0] = v_0$ in $G$ it holds that $\rav(\pi) \geq \rav(\pi')$.
 We can also define corresponding \emph{threshold paths}:  given a value $R$  a path $\pi$
 is said to be  threshold $R$-optimal  if  $\rsum^{(N)}(\pi) \geq R$ or $\rav(\pi) \geq R$, respectively.
An \emph{$\epsilon$-optimal} path is one which is 	$\Rsum^{(N)}(v_0) -\epsilon$ or $\Rav(v_0) - \epsilon$ threshold optimal (for finite or infinite horizon respectively).

 \begin{example}
 	\label{ex:4node-optimal}
 	Consider again the graph $G_{\mathsf e}$ shown in Figure \ref{fig:graph1}. Examining the expressions computed in Example~\ref{ex:4node-reward},
 	we have that $\rav(\pi_2) > \rav(\pi_3)$ for $\gamma=0.1$ and $\rav(\pi_3) > \rav(\pi_2)$ for $\gamma=0.9$.
 	Thus, in general, the optimal value depends on $\gamma$.
 	Due to the structure of the set of infinite paths in $G_{\mathsf e}$ we can analytically compute the optimal value $\Rav(v)$ for each $v \in V_{\mathsf e}$ as a function of $\gamma$ and a corresponding optimal path (the proof is in the appendix):
 	\begin{itemize}
 		\item if $\gamma \in [0,a_1]$, then $\Rav(v) = \frac{\lambda}{1-\gamma} - \frac{\lambda}{1-\gamma}\gamma^3$ and the path $(abc)^\omega$ is optimal;
 		\item if $\gamma \in [a_1,a_2]$, then $\Rav(v) = \frac{\lambda}{1-\gamma} - \frac{\lambda}{1-\gamma}\cdot\frac{(\gamma^2+4\gamma^3+2\gamma^5+\gamma^8)}{8}$ and 
 		$(abcabcad)^\omega$ is optimal;
 		\item if $\gamma \in [a_2,1]$, then $\Rav(v)\! = \! \frac{\lambda}{1-\gamma} \!- \!\frac{\lambda}{1-\gamma}\! \cdot\! \frac{(\gamma^2+\gamma^3+3\gamma^5\!)}{5}$ and the path $(abcad)^\omega$ is optimal.
 	\end{itemize}
 	The constants $a_1 \approx 0.2587$ and $a_2 \approx 0.2738$ are respectively the unique real roots of polynomials $\gamma^6+2\gamma^3-4\gamma+1$ and $5\gamma^6-14\gamma^3+12\gamma-3$ in the interval $(0,1)$.
 	Note that for $\gamma=1$ we have $\Rav(v) = 4\lambda$ which is achieved by $(abcad)^\omega$. The path $(abc)(ad)(abc)^2(ad)^2\hspace{-0.05in}\ldots(abc)^n(ad)^n\hspace{-0.05in}\ldots$, which is not ultimately periodic, also achieves the optimal reward.
 \end{example}
 
\vspace*{-7mm}
 \subparagraph*{Paths as strategies}
 We often refer to infinite paths as resulting from strategies (of the collecting agent).
 A \emph{strategy} $\sigma$ in $G$ is a function that maps finite paths $\pi[0\ldots m]$ to nodes such that if $\sigma(\pi)$ is defined then $(\pi[m],\sigma(\pi)) \in E$. Given an initial node $v_0$, the strategy $\sigma$ generates a unique infinite path $\pi$, denoted as $\outcome(v_0, \sigma)$.
 Thus, every infinite path $\pi = v_0,v_1,\ldots$ defines a unique strategy $\sigma_\pi$ where $\sigma_\pi(\pi[0\ldots i]) = v_{i+1}$, and $\sigma_\pi(\epsilon) = v_0$, and $\sigma_\pi$ is undefined otherwise. Clearly, $\outcome(v_0, \sigma_\pi) = \pi$.
 We say a strategy $\sigma$ is optimal for a path problem if the path 
 $\outcome(v_0, \sigma)$ is optimal.
 A strategy $\sigma$ is \emph{memoryless} if for every two paths $\pi'[0\ldots m'], \pi''[0\ldots m'']$ for which 
 $\pi'[m']=\pi''[m'']$, it holds that $\sigma(\pi') = \sigma(\pi'')$.
 We say that memoryless strategies suffice for the optimal path problem if there always exists
 a memoryless strategy $\sigma$ such that $\outcome(v_0, \sigma)$ is an optimal path.

\section{Finite Horizon Rewards: Computing $\Rsum^{(N)}(v)$}
\label{section:FiniteHorizon}
In this section we consider the finite-horizon problems associated with our model.
The following theorem summarizes the main results.

\begin{theorem}
	\label{thm:complexity-finite}
	Given $G = (V,\!E)$, expected reward and discounting functions, node $v\!\in\!V$, and 
	horizon $N\!\in\! \mathbb N$:
	\begin{compactenum}
		\item The finite-horizon value decision problem is NP-complete if 
		$N$ is in unary.
		\item 
		The value $\Rsum^{(N)}(v)$ for  $v\!\in\! V$  
		is computable in exponential time even if
		$N$  is in binary.
	\end{compactenum}
\end{theorem}

Analogous results hold for the related reward problem where
in addition to the initial node $v$, we are also given a destination node $v_f$, and
the objective is to go from $v$ to $v_f$ in at most $N$ steps while maximizing the reward.

The finite-horizon value problem is NP-hard by reduction from the Hamiltonian path problem
(the proof is in the appendix), even in the case of node-invariant $\lambda$ and $\gamma$.
Membership in NP in case $N$ is in unary follows from the fact that we can 
guess a path of length $N$ and check that the reward for that path is at least the
desired threshold value.

To prove the second part of the theorem, 
we  construct a finite \emph{augmented weighted graph}.

For simplicity, we give the proof for node-invariant $\lambda$ and $\gamma$, 
working with the cost functions $\csum$ and $\cav$. 
The augmented graph construction in the general case is a trivial generalization by changing the weights of the nodes, and the dynamic programming algorithm used for computing the optimal cost values is easily modified to compute the corresponding reward values instead.
For $\gamma < 1$ the objective is to minimize $\csum^{(N)}(\pi) = \sum_{t=0}^N \gamma^{\Last_\pi(t,\pi[t])}$ and for $\gamma = 1$ the objective is to maximize 
$\rsum^{(N)}(\pi) = \lambda\sum_{t=0}^{N}\Last_\pi(t,\pi[t])$ over paths $\pi$.

\vspace*{-4mm}
\subparagraph*{Augmented weighted graph}
Given a finite directed graph $G = (V,E)$ we define the \emph{augmented weighted graph} 
$\widetilde G = (\widetilde V, \widetilde E)$ 
which ``encodes'' the values $\Last_\pi(t,v) $ for the paths in $G$ explicitly
in the augmented graph node.
We can assume w.l.o.g.\ that $V= \set{1,2,\dots, \abs{V}}$.
\begin{compactitem}
	\item 
	The set of nodes $\widetilde V $ is $ V \times \mathbb Z_+^{\abs{V}}$ (the set
	$\widetilde V $ is infinite).
	A node $(v, b_1, b_2, \dots, b_{\abs{V}}) \in \widetilde V$ represents the fact that the current node is  $v$, and that for each node $u\in V$ the last visit to $u$ (before the current time) 
	was $b_u$ time units before the current time.
	\item The weight of a node  $(v, b_1, b_2, \dots, b_{\abs{V}}) \in \widetilde V$ 
	is $\gamma^{b_v}$.
	\item The set of edges $\widetilde E $ consists of edges
	$(v, b_1, b_2, \dots, b_{\abs{V}}) \ \rightarrow\  (v', b'_1, b'_2, \dots, b'_{\abs{V}})$ such that
	$(v,v') \in E$; and
	$b'_v = 1$, and $b'_u= b_u+1$ for all $u\neq v$.
\end{compactitem}
Let $\pi$ be a path in $G$. 
In the graph $\widetilde G$ there exists a unique path $\widetilde \pi$ that corresponds to $\pi$: 
\begin{equation}
\label{eq:aug_path}
\widetilde \pi = (\pi[0], 1, 1, \dots, 1)\, ,\, (\pi[1], b^1_1, b^1_2, \dots, b^1_{\abs{V}})\, ,\,
(\pi[2], b^2_1, b^2_2, \dots, b^2_{\abs{V}})\, , \dots
\end{equation}
starting from the node $(\pi[0], 1, 1, \dots, 1)$
such that for all $t$ and for all $v\in V$, we have $\Last_\pi(t,v) = b^t_{v}$.
Dually, for each path $\widetilde \pi $ in $\widetilde G$ starting from $(v_0, 1, 1, \dots, 1)$,
there exists a unique path $\pi$ in $G$ from the node $v_0$ such that $\Last_\pi(t,v) = b^t_{v}$ for all $t$ and $v$.

For a path $\widetilde \pi$ in the form of \eqref{eq:aug_path}
let
\begin{equation*}
{\widetilde c}_{\ssum}^{(N)} (\widetilde \pi)  :=   \sum_{t=0}^N \gamma^{b^t_{v_t}}\quad \text{ and }\quad
{\widetilde c}_{\sav}(\widetilde \pi)  :=   \limsup_{N \to \infty}\frac{1}{N+1}
\sum_{t=0}^N \gamma^{b^t_{v_t}}. \text{ Observe that:}
\end{equation*}
\begin{compactitem}
	\item ${\widetilde c}_{\ssum}^{(N)} (\widetilde \pi) $ is the sum of weights associated with the first $N+1$ nodes of $\widetilde \pi$.
	\item ${\widetilde c}_{\sav}(\widetilde \pi) $ is the limit-average of the weights associated with the nodes of $\widetilde \pi$.
\end{compactitem}
Thus, ${\widetilde c}_{\ssum}^{(N)}$ and 
${\widetilde c}_{\sav}$ 
define the classical total finite sum (shortest paths) and limit average objectives on weighted (infinite) graphs~\cite{ZwickP96}.
Additionally,
$  {\widetilde c}_{\ssum}^{(N)} (\widetilde \pi)  \ =  \csum^{(N)} ( \pi)$, and
$ {\widetilde c}_{\sav} (\widetilde \pi) \ = \cav( \pi)$ where $\pi$ is the path in $G$ corresponding to
the path $\widetilde \pi$.

Now, define ${\widetilde C}_{\ssum}^{(N)} 
\left((v_0, 1,1,\dots, 1) \right)$ as the infimum of ${\widetilde c}_{\ssum}^{(N)}(\widetilde \pi)$ over all paths
$\widetilde \pi$ with $\widetilde{\pi}[0] =  (v_0, 1,1,\dots, 1)$, and similarly for
$  {\widetilde C}_{\sav}\left((v_0, 1,1,\dots, 1) \right)$.
Then it is easy to see that 
$ \Csum^{(N)} ( v_0) ={\widetilde C}_{\ssum}^{(N)}  \left((v_0, 1,1,\dots, 1) \right)$
and
$\Cav( v_0) =  {\widetilde C}_{\sav}  \left((v_0, 1,1,\dots, 1) \right)$.
Thus, we can reduce the optimal path and value problems for $G$ to standard objectives in  $\widetilde G$. The major difficulty is that $\widetilde G$ is infinite. 
However, note that only the first $N+1$ nodes of $\widetilde \pi$ are relevant for the computation of  $ {\widetilde c}_{\ssum}^{(N)}   (\widetilde \pi)$.
Thus, the value of $ {\widetilde C}_{\ssum}^{(N)} ( (v_0, 1, \ldots, 1)) $ can be computed on a \emph{finite}
subgraph of $\widetilde G$, obtained by considering only the finite subset of
nodes $V\!\times \! \mathbb Z[1,N+1]^{\abs{V}} \ \subseteq \widetilde V$.

For $\gamma < 1$, we  obtain the value ${\widetilde C}_{\ssum}^{(N)}  (\widetilde \pi)$ by a standard dynamic programming algorithm which
computes the shortest path of length $N$ on this finite subgraph
starting from the node $(v_0, 1, 1,\dots, 1)$ (and keeping track of the number of steps).
For $\gamma = 1$, where the objective is to maximize $\rsum^{(N)}(\pi) = \lambda\sum_{t=0}^{N}\Last_\pi(t,\pi[t])$ over paths $\pi$, we proceed analogously.
Note that the subgraph used for the dynamic programming computations is of exponential size in terms of the size of $G$
and the description of $N$.
This gives the desired result in Theorem \ref{thm:complexity-finite}.

\section{Infinite Horizon Rewards: Computing $\Rav(v)$}
\label{section:InfiniteHorizon}
Since we consider finite graphs, every infinite path eventually stays in some strongly 
connected component (SCC). 
Furthermore, the value of the reward function $\rav(\pi)$ does not change if 
we alter/remove a finite prefix of the path $\pi$. 
Thus, it suffices to restrict our attention to the SCCs of the graph: 
the problem of finding an optimal path from a node $v \in V$ reduces to 
finding the SCC that gives the highest reward among the SCCs reachable from node $v$. 
Therefore, we  assume that the graph is strongly connected.

\subsection{Hardness of Exact $\Rav(v)$ Value Computation}

Since it is sufficient for the hardness results, we consider node-invariant $\lambda$ and $\gamma$.

\vspace*{-4mm}
\subparagraph*{Insufficiency of memoryless strategies.} 
Before we turn  to the computational hardness of the value decision problem, 
we look at the \emph{strategy complexity} of the optimal path problem and show that optimal strategies need memory.

\begin{proposition}
	Memoryless strategies do not suffice for the infinite horizon 
	problem.\looseness=-1
\end{proposition}
\begin{proof}
	Consider Example~\ref{ex:4node-optimal}.
	A memoryless strategy results in paths which cycle exclusively either in the
	left cycle, or the right cycle (as from node $a$, it prescribes a move to either
	only  $b$ or only $d$).
	As shown in Example~\ref{ex:4node-optimal}, the optimal path for $\gamma\!\ge\! a_1$
	needs to visit both cycles.
	Thus, memoryless strategies do not suffice for this example.
\end{proof}

For $\omega$-regular objectives, strategies based on \emph{latest visitation records}
\cite{GH82,Thomas95}, which depend only on the \emph{order} of 
the last node visits (\ie, for all node pairs $v_1\neq v_2 \in V$, whether the last visit of $v_1$ was before that of $v_2$ or vice versa)
are sufficient.
However, we can show that such strategies do not suffice either.
To see this, recall the  graph in Figure~\ref{fig:graph1} for which the optimal path for $\gamma=0.26$ is $(abcabcad)^\omega$. Upon visiting node $a$ this strategy chooses one of the nodes $b$ or $d$ depending on the \emph{number} of visits to $b$ since the last occurrence of $d$. On the other hand, every strategy based \emph{only on the order of last visits} is not able to count the number of visits to $b$ and thus, results in a path that ends up in one of $(abc)^\omega$, or $(ad)^\omega$, or $(abcad)^\omega$, which are not optimal for this $\gamma$. The proof is given in the appendix.
It is open if finite memory strategies are sufficient for the infinite-horizon optimal path problem.

\vspace*{-4mm}
\subparagraph*{NP-Hardness of the value decision problem}
To show NP-hardness of the
infinite-horizon value decision problem,
we first give bounds on $\Rav(v_0)$. The following Lemma proven in the appendix, establishes these bounds.
\begin{lemma}
	\label{lemma:Hamilt}
	For any graph $G = (V,E)$ and any node $v_0\in V$, $\Rav(v_0)$ is bounded as
	\begin{equation}
	\label{eq:bounds}
	\lambda\frac{1-\gamma^p}{1-\gamma}\le \,\Rav(v_0)\,\le \lambda\frac{1-\gamma^{n_{\mathsf v}}}{1-\gamma},
	\end{equation}
	where $n_{\mathsf v} = |V|$ and $p$ is the length of the longest simple cycle in $G$.
\end{lemma}

\begin{corollary}
	\label{cor:Hamilt}
	If the graph $G = (V,E)$ contains a Hamiltonian cycle, 
	any path $\pi = (v_1v_2\ldots v_{|V|})^\omega$, with $v_1v_2\ldots v_{|V|}v_1$ being a Hamiltonian cycle,
	is optimal and the optimal value of $\Rav(v_0)$ is exactly the upper bound in \eqref{eq:bounds}.
\end{corollary}

The following lemma establishes a relationship between the value of optimal paths and the existence of a Hamiltonian cycle in the graph, and is useful for providing a lower bound on the computational complexity of the value decision problem.

\begin{lemma}
	\label{lemma:existHamilt}
	If the upper bound in \eqref{eq:bounds} is achieved with a path $\pi$
	for some $\gamma< 1/|V|$, then the graph contains a Hamiltonian cycle $\rho$ that occurs in $\pi$ infinitely often.
\end{lemma}
\begin{proof}
	The proof is by contradiction. Suppose $\pi$ does not visit any Hamiltonian cycle infinitely often. Then it visits each such cycle at most a finite number of times.
	Without loss of generality we can assume that the path doesn't visit any such cycles, since the total number of Hamiltonian cycles is finite. We have for $n_{\mathsf v}=|V|$
	\begin{align*}
	\csum^{(N)}(\pi) &\ge \sum_{t=0}^N \gamma^{\Last_\pi(t,\pi[t])}\mathbb I(\Last_\pi(t,\pi[t])<n_{\mathsf v})
	\ge \gamma^{n_{\mathsf v}-1}\sum_{t=0}^N\mathbb I(\Last_\pi(t,\pi[t])<n_{\mathsf v}).
	\end{align*}
	Now let's look at $\pi_{n,n_{\mathsf v}} = \pi[n(n+1)\ldots(n+n_{\mathsf v})]$ a finite sub-path of length $n_{\mathsf v}$. There is at least one node repeated in $\pi_{n,n_{\mathsf v}}$ since the graph has $n_{\mathsf v}$ distinct nodes. Note that if $\pi[n]=\pi[n+n_{\mathsf v}]$, there must be another repetition due to the lack of Hamiltonian cycles in the path.
	In either case, there is an $i_n\in\mathbb Z[n,n+n_{\mathsf v}]$ such that $\Last_\pi(i_n,\pi[i_n])<n_{\mathsf v}$.
	\begin{align*}
	& \csum^{(N)}(\pi) \ge \gamma^{n_{\mathsf v}-1}\sum_{t=0}^N\mathbb I(\Last_\pi(t,\pi[t])<n_{\mathsf v})\ge \gamma^{n_{\mathsf v}-1}\left\lfloor\frac{N}{n_{\mathsf v}}\right\rfloor\\
	& \Rightarrow \cav(\pi)\ge \gamma^{n_{\mathsf v}-1}\limsup_{N\rightarrow\infty}\frac{1}{N+1}\left\lfloor\frac{N}{n_{\mathsf v}}\right\rfloor
	= \gamma^{n_{\mathsf v}-1}\frac{1}{n_{\mathsf v}}
	\Rightarrow \gamma^{n_{\mathsf v}} \ge 
	\left(\gamma^{n_{\mathsf v}-1}\right)/n_{\mathsf v} \ \Rightarrow\  \gamma \ge 1/n_{\mathsf v},
	\end{align*}
	which is a contradiction with the assumption that $\gamma< 1/n_{\mathsf v}$.
	Then a necessary condition for $\cav(\pi) \!= \!\gamma^{n_{\mathsf v}}$ with some 
	$\gamma\!<\!1/n_{\mathsf v}$ is the existence of a Hamiltonian cycle.
\end{proof}
\noindent\emph{Remark.}
Following the same reasoning of the above proof it is possible to improve the upper bound in \eqref{eq:bounds} as
$\Rav(v_0)\le \lambda\left(1-\gamma^p/n_{\mathsf v}\right)/(1-\gamma)$
for small values of $\gamma$, where $p$ is the length of the longest simple cycle of the graph.

\begin{theorem}\label{thm:complexity-infinite}
	The infinite horizon value decision problem is NP-hard for $\gamma\! < \!1/|V|$.
\end{theorem}
\begin{proof}
	We reduce the Hamiltonian cycle problem to the infinite horizon optimal path problem. 
	Given a graph $G = (V,E)$, we fix some $\lambda$ and $\gamma < 1/|V|$. We show that $G$ is Hamiltonian iff
	$\Rav(v_0) \geq \lambda\left(1-\gamma^{|V|}\right)/(1-\gamma)$.
	If $G$ has a Hamiltonian cycle $v_1v_2\ldots v_{|V|}v_1$, then the infinite path $\pi = (v_1v_2\ldots v_{|V|})^\omega$ has reward $\rav(\pi) =\lambda\left(1-\gamma^{|V|}\right)/(1-\gamma)$, for any choice of $\gamma$.
	For the other direction, applying Lemma~\ref{lemma:existHamilt} with 
	$\gamma < 1/|V|$ implies that $G$ is Hamiltonian.
\end{proof}

\vspace*{-4mm}
\subparagraph*{Non-discounted rewards ($\gamma = 1$) and node-invariant function $\lambda$}
Contrary to the finite-horizon non-discounted case, the infinite-horizon optimal path and value problems for $\gamma=1$ can be solved in polynomial time. To see this, note that the reward expression $\rsum^{(N)}(\pi)$ can be written as $\rsum^{(N)}(\pi) = \lambda\sum_{v\in V}y(v,\pi,N)$,
where $y(v,\pi,N)$ is defined as $y(v,\pi,N) = 1+\max U,\,\,\text{for } U = \{j\in \mathbb{N}|j\le N, \pi[j]=v\}\cup\{-1\}$. 
Then, we can bound the reward by\looseness=-1
\hspace{-3mm}\begin{align*}
\rav(\pi)&\le \lim_{N\rightarrow\infty}\frac{\lambda \sum_{i=1}^{|V(\pi)|}(N+1-i+1)}{N+1} = \lim_{N\rightarrow\infty}\frac{\lambda |V(\pi)|(2N-|V(\pi)|+3)}{2(N+1)}
= \lambda|\Inf(\pi)|,
\end{align*}
where $\Inf(\pi)$ is the set of nodes visited in the path $\pi$ infinitely often.
This indicates that the maximum reward is bounded by $\lambda$ times the maximal size of a reachable SCC in the graph $G$. This upper bound is also achievable: we can construct an optimal path by finding a maximal SCC reachable from the initial node and a (not necessarily simple) cycle $v_1\ldots v_n v_1$ that visits all the nodes in this SCC. Then, a subset of optimal paths contains paths of the form $\pi_0\cdot(v_1\ldots v_n)^\omega$, where $\pi_0$ is any finite path that reaches $v_1$.
This procedure can be done with a computational time polynomial in the size of $G$.

\begin{example}
	Consider the graph in Figure \ref{fig:graph1}. The optimal reward for the infinite-horizon non-discounted case is $4\lambda$, achievable by the path $\pi = (abcad)^\omega$. Another path which is not ultimately periodic but achieves the same reward
	$\pi' = (abc)(ad)(abc)^2(ad)^2\ldots(abc)^n(ad)^n\ldots$.
\end{example}

Note that for the case $\gamma=1$ there always exists an ultimately periodic optimal path, such a path is generated by a finite-memory strategy.

\subsection{Approximate Computation of $\Rav(v)$}

In the previous section we discussed how to solve the infinite-horizon value and path computation problems for the non-discounted case.
Now we show how the infinite-horizon path and value computation problems for $\gamma < 1$ can be effectively approximated. We first define functions that over and underapproximate $\Cav(v)$ (thus also $\Rav(v)$) and establish bounds on the error of these approximations.
Given an integer $K \in \mathbb{N}$, approximately optimal paths and an associated interval approximating $\Cav(v)$ can be computed using a finite augmented graph 
$\widetilde G_K$ based on the augmented graph $\widetilde G$ of
Section~\ref{section:FiniteHorizon}. Intuitively, $\widetilde G_K$  is obtained from $\widetilde G$ by pruning nodes that have a component greater than $K$ in their augmentation.
By increasing the value of $K$, the approximation error can be made arbitrarily small. 
\looseness=-1

We describe the approximation algorithm for node-invariant $\lambda$ and $\gamma$. 
The results generalize trivially to the case when $\lambda$ and $\gamma$ are not node-invariant by choosing 
$K$ large enough to satisfy the condition that bounds the approximation error for each $\lambda(v)$ and $\gamma(v)$.

\vspace*{-4mm}
\subparagraph*{Approximate cost functions}
Consider the following  functions from $V^*\!\!\times\!\mathbb N$ to $\mathbb R_{\geq 0}$:
\begin{align*}
\ocsum^{(N)}(\pi,K) & := \sum_{t=0}^N \gamma^{\min\{K,\Last_\pi(t,\pi[t])\}} \quad \text{ and }\\
\ucsum^{(N)}(\pi,K) & := \sum_{t=0}^N \gamma^{\Last_\pi(t,\pi[t])}\mathbb I(\Last_\pi(t,\pi[t])\le K).
\end{align*}
Informally, for $\ocsum^{(N)}(\pi,K)$, if the last visit to node $\pi[t]$ occurred
more than $K$  time units before time $t$, the cost is $\gamma^K$, rather
than the  original smaller amount $ \gamma^{\Last_\pi(t,\pi[t])}$.
For $\ucsum^{(N)}(\pi,K)$, if the last visit to $\pi[t]$ occurred
more than $K$  time steps before time $t$, then the cost is $0$.
For both, if the last visit to the node $\pi[t]$ occurred
less than or equal to $K$ steps before, we pay the actual cost
$\gamma^{\Last_\pi(t,\pi[t])}$.
The above definition implies that $\ucsum^{(N)}(\pi,K) \le \csum^{(N)}(\pi)\le \ocsum^{(N)}(\pi,K)$ for every $\pi$.
Then we have $\underline C(v_0,K) \leq \Cav(v_0) \leq  \overline C(v_0,K)$, where we define
\begin{align*}
\underline C(v_0,K)& :=
\inf_{\pi,\pi[0] = v_0}\limsup_{N \to \infty}\frac{\ucsum^{(N)}(\pi,K)}{N+1} \quad\text{ and }\\
\overline C(v_0,K) & :=
\inf_{\pi,\pi[0] = v_0}\limsup_{N \to \infty}\frac{\ocsum^{(N)}(\pi,K)}{N+1}.
\end{align*}
The difference between the upper and lower bounds can be tuned by selecting $K$:
\begin{align*}
& \ocsum^{(N)}(\pi,K)- \ucsum^{(N)}(\pi,K) = \sum_{t=0}^N \gamma^K\mathbb I(\Last_\pi(t,\pi[t])\ge K+1)\\
& \Longrightarrow 0\le \ocsum^{(N)}(\pi,K)- \ucsum^{(N)}(\pi,K) \le (N+1)\gamma^K\\
& \Longrightarrow \ucsum^{(N)}(\pi,K)\le \ocsum^{(N)}(\pi,K)\le \ucsum^{(N)}(\pi,K) + (N+1)\gamma^K\\
& \Longrightarrow \underline C(v_0,K)\le \overline C(v_0,K) \le \underline C(v_0,K)+\gamma^K.
\end{align*}
Therefore $\Cav(v_0)$ belongs to the interval $[\underline C(v_0,K),\overline C(v_0,K)]\subset[0,\gamma^K]$ and the length of the interval is at most $\gamma^K$. In order to guarantee the total error of $\epsilon>0$ for the actual reward $\Rav(v_0)$\footnote{Since $\Rav(v_0)$ is upper bounded by $\lambda/(1-\gamma)$, we assume that the required accuracy $\epsilon$ is less than this upper bound.}, we can select $K\in \mathbb N$ according to
$
\frac{\lambda}{1-\gamma}\gamma^K\le\epsilon\Longrightarrow
K\ge \ln\left[\frac{\epsilon(1-\gamma)}{\lambda}\right]/\ln\gamma.
$
The value $\Cav(v_0)$ can be computed up to the desired degree of accuracy by
computing either $\overline C(v_0,K)$ or $\underline C(v_0,K)$.
Next, we give the procedure for computing 
$\overline C(v_0,K)$ (the procedure for $\underline C(v_0,K)$ is similar).
It utilizes a finite augmented weighted graph $\widetilde G_K$.

\vspace*{-4mm}
\subparagraph*{Truncated augmented weighted graph $\widetilde G_K$}
Recall the infinite augmented weighted graph $\widetilde G$ from 
Section~\ref{section:FiniteHorizon}.
We define a truncated version $\widetilde G_K= (\widetilde V_K, \widetilde E_K)$ 
of  $\widetilde G$  where
we only keep track of   $\Last_\pi(t,v)$ values less than or equal to $K$.
For $\widetilde G_K$ we define two weight functions 
$\overline w$ and $\underline w$, for $\ocsum^{(N)}$  and $\ucsum^{(N)}$ respectively.
\begin{compactitem}
	\item 
	The set of nodes $\widetilde V_K $ is $ V \times \mathbb   Z[0,K]^{\abs{V}}$.
	\item For a node $\widetilde v = (v, b_1, b_2, \dots, b_{\abs{V}}) \in \widetilde V$,
	\begin{itemize}
		\item the weight $\overline\omega (\widetilde v)$ is $\gamma^{b_v}$ if $b_v > 0$ and $\gamma^{K}$ otherwise.
		\item the weight $\underline{\omega} (\widetilde v)$ is $\gamma^{b_v}$ if $b_v > 0$ and $0$ otherwise.
	\end{itemize}
	\item The set of edges $\widetilde E_K $ consists of edges
	$(v, b_1, b_2, \dots, b_{\abs{V}}) \ \rightarrow\  (v', b'_1, b'_2, \dots, b'_{\abs{V}})$ such that
	$(v \rightarrow v') \in E$,
	$b'_v = 1$, and for $u\neq v$:
	\begin{compactitem}
		\item $b'_u= b_u+1$ if  $b_u > 0$ and $b_u+1 \leq K$,
		\item $b'_u=0$  if $b_u=0$ or  $b_u+1 > K$.
	\end{compactitem}
	Thus, in $\widetilde G_K$ we specify two different weights for path $\pi$ at time step $t$ in the case when the previous visit to $\pi[t]$ in $\pi$ was more than $K$ time steps ago.
\end{compactitem}

Similarly to the infinite augmented graph we have 
\begin{compactitem}
	\item $\ocsum^{(N)}(\pi,K)$ is the sum of weights assigned by $\overline w$ to  the first $(N+1)$ nodes of $\widetilde \pi$.
	\item $\ucsum^{(N)}(\pi,K)$ is the sum of weights assigned by $\underline w$ to  the first $(N+1)$ nodes of $\widetilde \pi$.
\end{compactitem}
It is easy to see that $\overline C(v_0,K)$ is the least possible limit-average
cost with respect to $\overline w$ in $\widetilde G_K$ starting from the node $\widetilde v_0 = \left(v_0, 1,1,\dots, 1\right)$.
The same holds for $\underline C(v_0,K)$ with $\underline w$. Below we show how to compute  $\overline C(v_0,K)$. The case $\underline C(v_0,K)$ is analogous, and thus omitted.

\vspace*{-4mm}
\subparagraph*{Algorithm for computing $\overline C(v_0,K)$}
We now describe a method to compute $\overline C(v_0,K)$ as the least possible limit-average
cost in $\widetilde G_K$ with respect to $\overline w$. It is well-known that this can be reduced to the computation of the
minimum cycle mean in the weighted graph \cite{ZwickP96}, which in turn can be done using the algorithm from~\cite{Karp78} that we now describe.

As before, first we assume that $\widetilde G_K$ is strongly connected. 
For every $\widetilde v\in\widetilde V_K$, and every $n \in \mathbb Z_+$, we 
define $W_n(\widetilde v)$ as the minimum weight of a path of length $n$ from $\widetilde v_0$ to $\widetilde  v$; if no such path exists, then $W_n(\widetilde v) = \infty$. The values $W_n(\widetilde v)$ can be computed by the recurrence
\begin{equation*}
W_n(\widetilde v) = \min_{(\widetilde u, \widetilde v)\in\widetilde E_K}\{W_{n-1}(\widetilde u)+ \overline w(\widetilde v)\},\quad n = 1,2,\ldots,|\widetilde V_K|,
\end{equation*}
with the initial conditions $W_0(\widetilde v_0) = 0$ and 
$W_0(\widetilde v) = \infty$ for any $\widetilde v\ne \widetilde v_0$.
Then, we can compute 
\begin{equation*}
\overline C(v_0,K) = \min_{\widetilde v\in\widetilde V_K}\max_{n\in\mathbb Z[0,|\widetilde V_K|-1]}\left[\frac{W_{|\widetilde V_K|}(\widetilde v)-W_n(\widetilde v)}{|\widetilde V_K|-n} \right].
\end{equation*}

A cycle $\widetilde \rho$ with the computed minimum mean can be extracted by fixing the node $\widetilde v$ which achieves the minimum in the above value and the respective path length $n$ and finding a minimum-weight path from $\widetilde v_0$ to $\widetilde v$ and a cycle of length $|\widetilde V_K| - n$ within this path. 
Thus, the path $\widetilde \pi$ in $\widetilde G_K$  obtained by repeating $\widetilde \rho$ infinitely often realizes this value. A path $\pi$ from $v_0$ in $G$ with $\cav(\pi) = \overline C(v_0,K)$ is obtained from $\widetilde \pi$ by projection on $V$.

In the general case, when  $\widetilde G_K$ is not strongly connected we have to consider each of its SCCs reachable from $\widetilde v_0$, and determine the one with the least minimum cycle mean.

For each SCC with $m$ edges and $n$ nodes the computation of the quantities $W$ requires $O(n\cdot m)$ operations. The computation of the minimum cycle mean for this component requires $O(n^2)$ further operations. Since $n \leq m$ because of the strong connectivity, the overall computation time for the SCC is $O(n\cdot m)$. Finally, the SCCs of $\widetilde G_K$ can be computed in time $O(|\widetilde V_K| + |\widetilde E_K|)$ \cite{Tarjan71}. This gives us the following result.

\begin{lemma}\label{lem:approx-cost}
	The value $\overline C(v_0,K)$ and a path $\pi$ with limit average cost $\cav(\pi) = \overline C(v_0,K)$ can be computed in time polynomial in the size of $\widetilde G_K$.
\end{lemma}
The same result can be established for the under approximation $\underline C(v_0,K)$.

\noindent\emph{Remark.}
The number of nodes of $\widetilde G_K$ is $|\widetilde V_K| = |V| \times (K+1)^{|V|}$. 
For the approximation procedure described above it suffices to augment the graph with the information about which nodes were visited in the last $K$ steps and in what order. Thus, we can alternatively consider a graph with $|V| \times {|V|}^K$ nodes in the case when the computed  $K$ is smaller than $|V|$.

\begin{example}
	\label{ex:4node-approx}
	Figure \ref{fig:graph2} shows the SCC of the augmented graph $\widetilde G_K$ reachable from initial node $(a,1,1,1,1)$ for the graph given in Figure \ref{fig:graph1} and $K=5$.
	The nodes in the SCC are denoted by shorthands in the picture, \eg, $a_1 = (a,3,2,1,4)$.
	The labels on the nodes correspond to the values of the weight functions $\overline w$ and $\underline w$.
	As we can see, $\widetilde G_5$ already contains simple cycles $(a_2b_2c_2a_2),(a_3b_3c_2a_1b_1c_1a_2d_2a_3),(a_3b_3c_2a_1d_1a_3)$,
	which correspond to the optimal paths presented in Example~\ref{ex:4node-optimal}.
	The outcome of the minimum cycle mean for $\widetilde G_5$ will be the same with the two sets of weights only for the first and third interval for $\gamma$ determined in Example~\ref{ex:4node-optimal}, but will be different for the second case in which the term $\gamma^8$ is replaced respectively by $\textcolor{blue}{\gamma^5}$ and $\textcolor{red}{0}$ for the upper and lower bounds.
	\begin{figure}
		\begin{minipage}{0.55\textwidth}
			\includegraphics[scale=0.7]{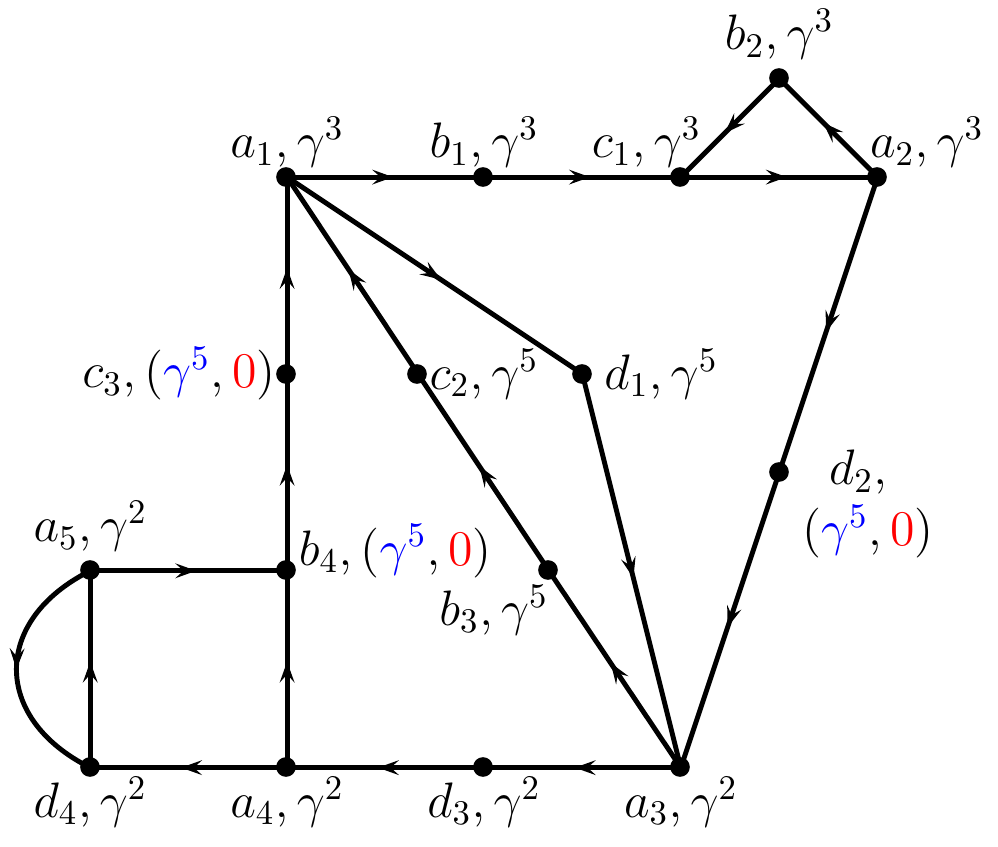}
		\end{minipage}
		\begin{minipage}{0.45\textwidth}
			$a_1 = (a,3,2,1,4)$
			\quad $a_2 = (a,3,2,1,0)$\\
			$a_3 = (a,2,4,3,1)$
			\quad $a_4 = (a,2,0,5,1)$\\
			$a_5 = (a,2,0,0,1)$\\[2mm]
			$b_1 = (b,1,3,2,5)$
			\quad $b_2 = (b,1,3,2,0)$\\
			$b_3 = (b,1,5,4,2)$
			\quad$ b_4 = (b,1,0,0,2)$\\[2mm]
			$c_1 = (c,2,1,3,0) $
			\quad $c_2 = (c,2,1,5,3) $\\
			$c_3 = (c,2,1,0,3)$\\[2mm]
			$d_1 = (d,1,3,2,5)$
			\quad $d_2 = (d,1,3,2,0)$\\
			$d_3 = (d,1,5,4,2)$
			\quad $d_4 = (d,1,0,0,2)$
		\end{minipage}
		\caption{The SCC of the finite augmented graph $\widetilde G_5$ for the graph in Figure~\ref{fig:graph1}. The node labels are the values of the functions $\overline w$ and $\underline w$ (in black if $\overline w = \underline w$, otherwise respectively in \textcolor{blue}{blue} and \textcolor{red}{red}).}
		\label{fig:graph2}
		\vspace{-.5cm}
	\end{figure}
\end{example}

Theorem~\ref{thm:approx-reward}, a consequence of Lemma~\ref{lem:approx-cost}, states the  approximate computation result.
\begin{theorem}\label{thm:approx-reward}
	Given a graph $G = (V,E)$, node $v_0 \in V$, rational values $\lambda$ and $0< \gamma < 1$, and error bound $\epsilon > 0$, we can compute in time polynomial in $|V|(K+1)^{|V|}$ for $K = \ln\left[\frac{\epsilon(1-\gamma)}{\lambda}\right]/\ln\gamma$ (i.e., exponential in $|V|$),
	infinite paths $\pi_{\mathsf{under}}$ and $\pi^{\mathsf{over}}$ and values $\rav(\pi_\mathsf{under})$ and $\rav(\pi_\mathsf{over})$ such that $\rav(\pi_\mathsf{under})\leq \Rav(v_0) \leq \rav(\pi_\mathsf{over})$ and $\rav(\pi_\mathsf{over}) -\rav(\pi_\mathsf{under})\le \epsilon$.
\end{theorem}

\subsection{Approximation via Bounded Memory}

The algorithm presented earlier is based on an augmentation of the graph with a specific structure updated deterministically and whose size depends on the desired quality of approximation. Furthermore, in this graph there exists a memoryless strategy with approximately optimal reward. 
We show that this allows us to quantify how far from the optimal reward value is 
a strategy that is optimal among the ones with bounded memory of fixed size. 

First, we give the definition of memory structures. 
A \emph{memory structure} $\mc M = (M,m_0,\delta)$ for a graph $G = (V,E)$ consists of a 
set $M$, initial memory $m_0 \in M$ and a memory update function $\delta :  M \times V \to M$.
The memory update function can be extended to $\delta^* :  V^* \to M$ by defining
$\delta^*(\epsilon) = m_0$ and $\delta^*(\pi\cdot v) = \delta(\delta^*(\pi),v)$.
A memory structure $\mc M$ together with a function $\tau : V \times M \to V$ such that $(v,\tau(v,m))\in E$ for all $v \in V$ and $m \in M$, and an initial node $v_0 \in V$ define a strategy $\sigma : V^* \to V$ where
$\sigma(\epsilon) = v_0$ and $\sigma(\pi\cdot v) = \tau(v,\delta^*(\pi))$.  
In this case we say that the strategy $\sigma$ has memory $\mc M$.
Given a bound $B \in \mathbb N$ on memory size we define the finite 
graph $G \times B = (V^{G\times B}, E^{G \times B})$, where 
$V^{G\times B} = V \times \mathbb Z[1,B]$; and
$E^{G \times B} = \{((v,i),(v',j)) \in (V\! \times\! \mathbb Z[1,B])\times (V\! \times\! \mathbb Z[1,B]) \mid (v,v') \in E\})$.

Memoryless strategies in this graph precisely correspond to strategies that have memory of size $B$. More precisely, for each strategy $\sigma$ in $G = (V,E)$ that has memory $\mc M = (M,m_0,\delta)$
there exists a memoryless strategy $\sigma_{\mc M}$ in $G \times |M|$ such that the projection of $\outcome((v_0,m_0),\sigma_{\mc M})$ on $V$ is $\outcome(v_0,\sigma)$. Conversely, for each memoryless strategy  $\sigma_{\mc M}$ in $G \times B$ there exist a memory structure $\mc M = (M,m_0,\delta)$ with $|M| = B$ and a strategy $\sigma$ with memory $\mc M$ in $G$ such that the projection of $\outcome((v_0,m_0),\sigma_{\mc M})$ on $V$ is $\outcome(v_0,\sigma)$.
\begin{example}
	Recall that in Example~\ref{ex:4node-optimal} we established that an optimal path for $\gamma=0.26$ is the path $(abcabcad)^\omega$. In the graph $G \times 3$ there exists a simple cycle corresponding to this path, namely
	$(a,1)(b,1)(c,2)(a,2)(b,2)(c,3)(a,3)(d,1)(a,1)$.
	Thus, the optimal path $(abcabcad)^\omega$ corresponds to a strategy with memory
	size of $3$.
\end{example}
An optimal strategy among those with memory of a given size $B$ can be computed by inspecting the  memoryless strategies in $G \times B$ and selecting the one with maximal reward (these strategies are finitely but exponentially many). 

A strategy returned by the approximation algorithm presented earlier uses a memory structure of size $(K+1)^{|V|}$. If $(K+1)^{|V|} \leq B$, then this strategy is isomorphic to some memoryless strategy $\sigma$ in $G \times B$. Since the reward achieved by the optimal memoryless strategy in $G \times B$ is at least that of $\sigma$,  its value is at most $\frac{\lambda}{1-\gamma}\gamma^{K}$ away from  $\Rav(v_0)$.
Now, 
$(K+1)^{|V|} \leq B$ iff $K \leq \left\lfloor\frac{\ln B}{\ln \abs{V}} \right\rfloor -1$.
Thus, under a  memory size $B$, we are guaranteed to find a strategy which
has reward that at most 
$\frac{\lambda}{1-\gamma}\gamma^{  \left\lfloor\frac{\ln B}{\ln \abs{V}} \right\rfloor -1  }$
away from the optimal.

Next we sketch an algorithm for computing optimal $B$-memory bounded strategies.
As mentioned previously, we can restrict our attention to memoryless strategies in
$G\times B$. Memoryless strategies in this graph lead to lasso shaped infinite paths,
with an initial non-cyclic path followed by a simple cycle which is repeated infinitely
often. This means that we can restrict our attention to paths of length
$\abs{V}\cdot B$ which complete the lasso.
Now, we apply this length bound to restrict our attention to the finite portion of the augmented
graph $\widetilde G$ from Section~\ref{section:FiniteHorizon} that corresponds
to path lengths at most $\abs{V}\!\cdot\! B$. 
This finite subgraph contains nodes $V\times \nat[1, \abs{V}\!\cdot\!B]^{\abs{V}}$.
The dynamic programming algorithm is polynomial time on this graph, hence
we get a running time which is polynomial in $\abs{V}\cdot \left(\abs{V}\cdot B\right)^{\abs{V}}$.

\begin{theorem}
	Given a graph $G = (V,E)$, a node $v_0 \in V$, rational values $\lambda$ and $0< \gamma < 1$, and a memory bound $B > 1$, we can compute a $B$-memory optimal strategy
	$\sigma$ with reward $\rav$  
	at most
	$\frac{\lambda}{1-\gamma}\gamma^{  \left\lfloor\frac{\ln B}{\ln \abs{V}} \right\rfloor -1  }$
	away from the optimal $\Rav(v_0)$ in time polynomial in $\abs{V}^{(\abs{V}+1)}\cdot B^{\abs{V}}$.
\end{theorem}

 \section{Generalizations of the Model}
 

\smallskip\noindent\textbf{General Decay Profiles.}
In this section we assume the  expected value of the generated reward for nodes
to be time independent.
Given a node $v\in V$,  the associated \emph{decay profile $\Gamma_v$} for the node
is defined to be  be a strictly monotonically decreasing
sequence $1, a_1, a_2, \dots$ (starting from $1$) converging to $0$.
We denote the $i$-th element of the sequence as $\Gamma_v(i)$, and the 
sequence portion $\Gamma_v(i), \Gamma_v(i+1), \dots \Gamma_v(j)$ as
$\Gamma_v[i..j]$.
If the original generated reward at a node $v$ was $\lambda$, after
$\Delta$ time units only $\lambda \cdot \Gamma_v(\Delta)$  of that reward remains after
decay if uncollected.
Denote the sum of the elements of $\Gamma_v[i..j]$ as $\mysum\left(\Gamma_v[i..j]\right)$,
that is 
$\mysum\left(\Gamma_v[i..j]\right) \ = \ \sum_{k=i}^j \Gamma_v(k)$.

\noindent For a path $\pi = v_0,v_1,\ldots,$ 
(and under fixed decay profiles $\Gamma_v$ for the nodes in $V$)  we define the finite horizon
cumulative (expected) reward as
\[
\rsum^{(N)}(\pi) \ =\  
 \sum_{t=0}^N \lambda(v_t) \cdot \mysum\left(\Gamma_{v_t}[0\, ..\  t\!-\! \Last_\pi(t,v_t)\! -\!1 ]\right).
\]
Note that the quantities 
$\lambda(v_t)\cdot \Gamma_{v_t}(0), \, \lambda(v_t)\cdot \Gamma_{v_t}(1), \ldots,
 \lambda(v_t)\cdot \Gamma_{v_t}(  t- \Last_\pi(t,v_t) -1 )$ are the expectations of the
uncollected rewards at node $v_t$ for times $t, t-1, \ldots,  t- \Last_\pi(t,v_t) -1$ respectively.

 We define the infinite horizon limit-average (expected) reward
as for the multiplicative discounting in \eqref{eq:rav-orig}.
Value decision problems are defined
analogously to Definition~\ref{def:ValueDecision}.

%
The following theorem shows that the $\rsum$-value decision problem remains NP-hard
under any decay profile (the proof is given in the appendix); and that the $\rsum$-values
can be computed in EXPTIME (a version of the augmented weighted graph can be defined similar to the multiplicative
discounting case to compute the finite horizon reward value $\Rsum^{(N)}(v)$).

\begin{theorem}
\label{theorem:GeneralFiniteNP}
Let  $G = (V,E)$ be a finite directed graph, and let the expected rewards and reward decay
profiles be node independent (and non-zero).
\begin{compactenum}
\item 
The finite horizon $\rsum$-value decision problem is NP-complete in case
$N$ is given in unary, and NP-hard otherwise.
\item The value $\Rsum^{(N)}(v)$ for  $v\!\in\! V$ (given $N$  in binary)
 is computable in EXPTIME.
\end{compactenum}
\end{theorem}
\begin{proof}
Consider $N = \abs{V}-1$.
Let $\beta  = \lambda\cdot \sum_{i=0}^N  \mysum\left(\Gamma[0\, ..\, i]\right)$.
We reduce deciding existence of the  Hamiltonian path to the  $\rsum$-value decision problem.
Clearly if there exists a Hamiltonian path in $G$, then the answer to the
$\rsum$-value decision problem  is affirmative.
Suppose the answer to the $\rsum$-value decision problem  is affirmative.
We show in this case that $G$ must have a Hamiltonian path.
Consider any non-Hamiltonian path $\pi= v_0, \dots, v_N$ of length $N$.
We show $\rsum(\pi) < \beta$.
We have 
\[
\frac{\beta-\rsum^{(N)}(\pi)}{\lambda} =
\sum_{i=0}^N a_i,
\]
where
\[
a_i  = \begin{cases}
& 0 \quad \text{ if  } v_i \text{ does not appear in } v_0, v_1, \dots, v_{i-1}\\
& \mysum\left(\Gamma[0 \,..\, i]\right) \ - \ 
 \mysum\left(\Gamma[0\, ..\ \Last_\pi(i,v_i) -1 ]\right)
\quad \text{ otherwise.}
\end{cases}
\]
Observe that \[\mysum\left(\Gamma[0 \,..\, i]\right) \ - \ 
 \mysum\left(\Gamma[0\, ..\  \Last_\pi(i,v_i) -1 ]\right) > 0\] in case
$\Last_\pi(i,v_i) < i+1$, \ie, if $v_i$ has appeared in the path
prefix $v_0, v_1, \dots, v_{i-1}$.
Since we must have at least one node which appears twice in $\pi$, we have that
at least one $a_i$ in $\sum_{i=0}^N a_i$ must be strictly positive.
This concludes the proof.
\end{proof}
In the infinite horizon case, a truncated weighted graph can be defined along the lines of
$\widetilde G_K$  to compute $\epsilon$-approximate bounds on $\Rsum^{(N)}(v)$
(\emph{even if } $\sum_{i=0}^{\infty} \Gamma_v(i) = \infty$).
In this case, given $\epsilon\! >\! 0$, the constant $K$ depends on the
rate of decay of the sequences $\Gamma_v$.

\smallskip\noindent\textbf{Two-Player Turn-Based Game Setting.}
In this setting, the nodes of the graph are partitioned into player-1 and
player-2 nodes, with the outgoing edges from a node 
being controlled by the respective player.
The objective of player~1 (the collecting agent) is to maximise its reward ($\rsum^N(\pi)$ or $\rav(\pi)$)
when facing an adversarial player~2 (as earlier, in this setting the collecting agent does
not receive information about the actual generated reward instances).
Plays and strategies are defined in the standard way (see \eg,~\cite{ZwickP96}).
It can be checked that the augmented weighted graphs of Sections~\ref{section:FiniteHorizon} and~\ref{section:InfiniteHorizon} also work in this game-based setting, and thus we obtain
algorithms for computing  $\Rsum^{(N)}(v)$ and $\epsilon$-optimal $\Rav(v)$ values applying algorithms
for shortest path and mean-payoff games~\cite{Darmann2017,ZwickP96}.
 
 \section{Conclusion and Open Problems}
 We have introduced the robot routing problem, which is a reward collection problem on graphs in which the 
 reward structure combines spatial aspects with dynamic and stochastic reward generation. We have studied the computational complexity of the finite and infinite-horizon versions of the problem, as well as the strategy complexity of the infinite-horizon case. We have shown that optimal strategies need memory in general, and that strategies based on last visitation records do not suffice. However, the important question about whether finite-memory suffices for the infinite-horizon optimal path problem or infinite memory is needed remains open. In case finite-memory suffices, it will be desirable to provide an algorithm for precisely solving the infinite-horizon value problem. So far we have only given methods for approximating the optimal solution.
 
 Another interesting line of future work is the generalization of the model to incorporate timing constraints. More specifically, in this work we have assumed that all the rewards accumulated at a node are instantaneously collected. This assumption is justified by the fact that in many request-serving scenarios the time taken to serve the requests at a given location is negligible compared to the time necessary to travel between the locations. A more precise model, however, would have to incorporate the serving time per node, which would depend on the amount of collected reward. This would imply that the elapsed time becomes a random variable, while in our current model it is deterministic. 
 
 Other  extensions include the setting where the robot can react to the environment by observing the collected rewards, or observing the accumulated rewards at nodes of the graph, or where there is  probabilistic uncertainty in the transitions of the robot in the graph.
 Finally, the robot routing problem presents new challenges for development of efficient approximation algorithms, which is a major research direction concerning path optimization problems.\looseness=-1


 \bibliography{main}
 
%
 \appendix
 \section*{Appendix}
 \section{Optimal Path for Example~\ref{ex:4node-optimal}}

Since the graph $G_{\mathsf e}$ is strongly connected, the optimal path is independent of the initial state. We can also neglect a finite prefix of any path with respect to the computation of the optimal limit-average collected rewards. Any path in the graph $G_{\mathsf e}$ can be (after neglecting a finite prefix) $(abc)^w$ or $(ad)^w$ or can be presented as $\pi = (abc)^{n_0}(ad)^{m_0}(abc)^{n_1}(ad)^{m_1}\ldots$ for unique sequences $\{n_k\}_k$ and $\{m_k\}_k$ where $n_k,m_k\in\mathbb Z_+$. In the following we focus on the limit-average cost function $\cav(\pi)$ and provide a path-independent lower bound which is also achievable by some paths in the graph.
To study the $\cav(\pi)$, we need to compute $\limsup_{N\rightarrow\infty}$ of the fraction $\csum^{(N)}(\pi)/(N+1)$.
Since we intend to provide a lower bound for this quantity, we focus on the limiting behavior of the fraction only for the subsequence $\{\sum_{k=0}^{n_{\mathfrak a}}(3n_k+2m_k)-1\ |\ n_{\mathfrak a} = 1,2,\ldots\}$.
The finite-horizon cost associated to $\pi$
for the horizon $N = \sum_{k=0}^{n_{\mathfrak a}}(3n_k+2m_k)-1$ is the following
\begin{align*}
\csum^{(N)}(\pi) = g_0(\gamma,n_0,m_0)& +\sum_{k=1}^{n_{\mathfrak a}}\left[\gamma^2+2\gamma^{2m_{k-1}+3}+(3n_k-3)\gamma^3\right]\\
 &+\sum_{k=1}^{n_{\mathfrak a}}\left[\gamma^3+\gamma^{3n_k+2}+(2m_k-2)\gamma^2\right],
\end{align*}
where $g_0(\gamma,n_0,m_0)$ is the polynomial that includes the terms generated by the prefix $(abc)^{n_0}(ad)^{m_0}$.
Then
\begin{align*}
\frac{\csum^{(N)}(\pi)}{N+1} &= \frac{\sum_{k=1}^{n_{\mathfrak a}} f(\gamma,n_k,m_k)}{\sum_{k=0}^{n_{\mathfrak a}}(3n_k+2m_k)}
+\frac{g_0(\gamma,n_0,m_0)+2\gamma^{2m_0+3}-2\gamma^{2m_n+3}}{\sum_{k=0}^{n_{\mathfrak a}}(3n_k+2m_k)}
\end{align*}
where $f(\gamma,n_k,m_k) := (2m_k-1)\gamma^2+(3n_k-2)\gamma^3+2\gamma^{2m_k+3}+\gamma^{3n_k+2}$.
The second fraction in the right-hand side goes to zero for $n_{\mathfrak a}\rightarrow\infty$ thus the limit of the fraction is the same as the following
\begin{align}
\frac{\sum_{k=1}^{n_{\mathfrak a}} f(\gamma,n_k,m_k)}{\sum_{k=1}^{n_{\mathfrak a}}(3n_k+2m_k)}
 = \frac{\sum_{i,j=1}a_{ij}^{(n_{\mathfrak a})}f(\gamma,i,j)}{\sum_{i,j=1}a_{ij}^{(n_{\mathfrak a})}(3i+2j)}
 = \sum_{i,j=1}\frac{a_{ij}^{(n_{\mathfrak a})}(3i+2j)}{\sum_{i_1,j_1=1}a_{i_1j_1}^{(n_{\mathfrak a})}(3i_1+2j_1)}\frac{f(\gamma,i,j)}{3i+2j},\label{eq:convex_comb}
\end{align}
where $a_{ij}^{(n_{\mathfrak a})}$ is the number of pairs $(i,j)$ that appear in the finite sequence $\{(n_k,m_k)\ |\ k=1,2,\ldots,n_{\mathfrak a}\}$.
This equality indicates that \eqref{eq:convex_comb} is a convex combination of fractions $f(\gamma,i,j)/(3i+2j)$ thus $\cav(\pi)$ is lower bounded by
\begin{equation}
\label{eq:lower_bound_cost}
\cav(\pi)\ge \inf\left\{\frac{f(\gamma,i,j)}{3i+2j}\,\bigg|\,i,j\in\mathbb Z_+\right\}.
\end{equation}
Note that $f(\gamma,i,j)/(3i+2j)$ is exactly the cost of path $\pi$ with the constant sequences $n_k=i$ and $m_k=j$ for all $k$.
Therefore, inequality \eqref{eq:lower_bound_cost} states that for the graph of Example~\ref{ex:4node-optimal}, the cost of non-periodic paths can not be less that the minimum cost of periodic paths. Then we need to compare a countable number of polynomials and find their minimum over the interval $\gamma\in(0,1)$. Such a comparison reveals that there are two support polynomials with $(i,j)$ equal to $(1,1),(2,1)$.
There is also one limit polynomial in which $j$ is fixed and $i$ goes to infinity. This polynomial $\gamma^3$ is not present in \eqref{eq:lower_bound_cost} explicitly but is generated by the path $(abc)^w$ which we left a side from the beginning. This completes the proof of computation in Example~\ref{ex:4node-optimal}.

\section{Missing Proofs from Section~\ref{section:FiniteHorizon}}
\medskip

The NP-hardness result (Theorem~\ref{thm:complexity-finite})
for $\gamma <1$ is based on the following lemma.
\begin{lemma}
For any finite graph $G = (V,E)$ with $|V| = n_{\mathsf v}$, $v\in V$ and $\gamma\in(0,1)$, we have that
\begin{equation}
\label{eq:Hamilt_finite}
\Rsum^{n_{\mathsf v}-1}(v) \le
\frac{\lambda\left(n_{\mathsf v}-(n_{\mathsf v}+1)\gamma+\gamma^{n_{\mathsf v}+1}\right)}{(1-\gamma)^2}.
\end{equation}
The equality holds if and only if the graph has a Hamiltonian path starting from $v$. 
\end{lemma}
\begin{proof}
The definition of $\Last_\pi(t,v)$ implies that $\Last_\pi(t,v)\le t+1$. Then for any path $\pi$ with $|\pi| = n_{\mathsf v}-1$ and $\pi[0]=v$,
\begin{align*}
\csum^{n_{\mathsf v}-1}(\pi) & = \sum_{t=0}^{n_{\mathsf v}-1}\gamma^{\Last_\pi(t,\pi[t])}\ge \sum_{t=0}^{n_{\mathsf v}-1} \gamma^{t+1} = \frac{\gamma(1-\gamma^{n_{\mathsf v}})}{1-\gamma}\\
\Rightarrow& \rsum^{n_{\mathsf v}-1}(\pi) \le  \frac{n_{\mathsf v}\lambda}{1-\gamma}
-  \frac{\lambda}{1-\gamma} \frac{\gamma(1-\gamma^{n_{\mathsf v}})}{1-\gamma},
\end{align*}
which proves \eqref{eq:Hamilt_finite}. If the path $\pi$ is Hamiltonian, then the path does not contain any repeated node and $\Last_\pi(t,\pi[t]) = t+1$ thus $\pi$ is optimal and achieves the upper bound in  \eqref{eq:Hamilt_finite}.
For the other direction, we show that a path $\pi$ that satisfies the condition of the claim is a Hamiltonian path.
Suppose that this is not the case. Since the length of $\pi$ is $n_{\mathsf v}-1$, $\pi$ contains at least one repeated node.
Thus for at least one index $t$ we have  $\Last_\pi(t,\pi[t])< t+1$, and then $\sum_{t=0}^{n_{\mathsf v}-1}\gamma^{\Last_\pi(t,\pi[t])}> \sum_{t=0}^{n_{\mathsf v}-1} \gamma^{t+1}$. Therefore, $\rsum^{n_{\mathsf v}-1}(\pi)$ is strictly less than the right-hand side of \eqref{eq:Hamilt_finite}, which contradicts the choice of the path $\pi$.
\end{proof}

\begin{proof}[\textbf{Theorem~\ref{thm:complexity-finite}}]
%
%
The case for $\gamma\in (0,1)$ can be proved as follows.
We reduce the Hamiltonian path problem to the finite horizon value decision problem as follows. 
Given a graph $G = (V,E)$ with $|V| = n_{\mathsf v}$, we fix some constants $\lambda$ and $\gamma$. By the inequality we showed in \eqref{eq:Hamilt_finite}, $G$ contains a Hamiltonian path if and only if for some $v \in V$ we have $\Rsum^{n_{\mathsf v}-1}(v) \geq \lambda\left[n_{\mathsf v}-(n_{\mathsf v}+1)\gamma+\gamma^{n_{\mathsf v}+1}\right]/(1-\gamma)^2$. This completes the reduction.

The proof for $\gamma =1$ is similar to the proof of 
Theorem~\ref{theorem:GeneralFiniteNP}.
\end{proof}

\section{Missing Proofs from Section~\ref{section:InfiniteHorizon}}
\noindent
\textbf{Lemma~\ref{lemma:Hamilt}}
For any graph $G = (V,E)$ and any node $v_0\in V$, $\Rav(v_0)$ is bounded as
\begin{equation}
\label{eq:boundsAppendix}
\lambda\frac{1-\gamma^p}{1-\gamma}\le \,\Rav(v_0)\,\le \lambda\frac{1-\gamma^{n_{\mathsf v}}}{1-\gamma},
\end{equation}
where $n_{\mathsf v} = |V|$ and $p$ is the length of the longest simple cycle in $G$.
\begin{proof}
Let $\pi$ be an infinite path and denote with $\pi_N = \pi[0\ldots N]$ the prefix of $\pi$ of length $N$. 
We apply the inequality of arithmetic and geometric means for $n\in\mathbb N$
\begin{equation*}
\frac{a_1\!+\!a_2\!+\!\ldots\!+\!a_n}{n}\ge \sqrt[n]{a_1a_2\ldots a_n}, \quad a_1,a_2,..,a_n\ge 0,
\end{equation*}
where the equality holds if and only if all $a_i$'s are equal.
Then we have
\begin{align*}
\csum^{(N)}(\pi) &\ge (N\!+\!1)\gamma^{\frac{1}{N+1}\sum_{t=0}^{N}\Last_\pi(t,\pi[t])}
= (N\!+\!1)\gamma^{\sum_v y_v/(N\!+\!1)},
\end{align*}
where $y_{v} = 1+\max\{j\in\mathbb N \mid j\le N,\,\, \pi[j] = v\}$ if $v$ is visited by $\pi_N$, otherwise $y_v = 0$.
This implies that $y_{v}\le N+1$ and
$
\csum^{(N)}(\pi)\ge (N+1)\gamma^{n_{\mathsf v}(N+1)/(N+1)} = (N+1)\gamma^{n_{\mathsf v}}.
$
Dividing both sides by $(N+1)$ and taking $\limsup_{N\rightarrow\infty}$ gives a general path-independent lower bound $\gamma^{n_{\mathsf v}}$ for $\cav(\pi)$,
which also implies $\Cav(v_0) =
$ $ \inf_{\pi}\left\{\cav(\pi)\ |\ \pi[0] = v_0\right\} \ge \gamma^{n_{\mathsf v}}$.
This proves the upper bound in \eqref{eq:boundsAppendix}.
The particular path $\pi_p = (v_1v_2\ldots v_p)^\omega$, with $v_1v_2\ldots v_pv_1$ being the longest simple cycle, gives the bound $\Rav(v_0) \ge \rav(\pi_p)$ which is the lower bound in \eqref{eq:boundsAppendix}.  
\end{proof}

\begin{proposition}
Strategies based only on last visitation records not suffice for the infinite-horizon optimal path problem.
\end{proposition}
\begin{proof}
Recall that for the graph $G_{\mathsf e} = (V_{\mathsf e},E_{\mathsf e})$ in Figure \ref{fig:graph1} and 
$\gamma=0.26$ the path $(abcabcad)^\omega$ is optimal. Upon visiting node $a$ the corresponding strategy chooses one of the nodes $b$ or $d$ depending on the \emph{number} of visits to $b$ since the last occurrence of $d$. We will now show that every strategy based \emph{only on the order of last visits} results in an ultimately periodic path the loop of which is one of the cycles $(abc)^\omega$, or $(ad)^\omega$, or $(abcad)^\omega$, and is hence not optimal.

To this end, consider the graph $\widehat G_{\mathsf e} = (\widehat V_{\mathsf e},\widehat E_{\mathsf e})$ with set of states $\widehat V_{\mathsf e} = V_{\mathsf e} \times O$, where the set $O$  consists of the vectors  of elements of $V$, of length less than or equal to $4$, that represent the order of last occurrences of the nodes $V$ in the traversed path. For example, $\langle a,c,b\rangle$ denotes the fact that the last occurrence of $a$ was after the last occurrence of $c$, which was after the last occurrence of $b$, and node $d$ has not occurred yet in the path. The edge relation $\widehat E_{\mathsf e} \subseteq \widehat V_{\mathsf e} \times \widehat V_{\mathsf e}$, where $(v,o) \rightarrow (v',o')$ implies $(v,v') \in E_{\mathsf e}$, updates the last occurrence order in the expected way, depending on the taken edge $(v,v')$. 

For example, the unique path in $\widehat G_{\mathsf e}$ that corresponds to $(abcabcad)^\omega$ is
\[\widehat\pi = (a,\langle\rangle)(b,o_1)(c,o_2)(a,o_3)(b,o_4)(c,o_5)(a,o_3)(d,o_4) \cdot \widehat{ \rho}^\omega,\]
where
\[\widehat{ \rho} = (a,o_6)(b,o_7)(c,o_8)(a,o_9)(b,o_{10})(c,o_{11})(a,o_{9})(d,o_{10})\]
and
\[
\begin{array}{l}
o_1  =  \langle a \rangle, 
o_2  =  \langle b,a\rangle, 
o_3  =  \langle c, b, a\rangle, 
o_4  =  \langle a,c,b \rangle,
o_5  =  \langle b,a,c \rangle, 
o_6  =  \langle  d,a,c,b\rangle, 
o_7 =  \langle  a,d,c,b\rangle,\\
o_8 =  \langle  b,a,d,c\rangle, 
o_9  =  \langle  c,b,a,d\rangle, 
o_{10}  =  \langle a,c,b,d \rangle,
o_{11}  =  \langle b,a,c,d \rangle.
\end{array}
\] 

Since there are occurrences of node $(a,o_3)$ in the path followed by different nodes, $(b,o_4)$ and $(d,o_4)$ respectively, and similarly for node $(a,o_9)$, followed by either $(b,o_{10})$ or $(d,o_{10})$, the path $\widehat\pi$ cannot be generated by a memoryless strategy in $\widehat G$.

In general, in $\widehat G$ there are $5$ possible reachable states of the form $(a,o)$, namely,
$(a,\langle\rangle)$, $(a,\langle c,b,a\rangle)$, $(a,\langle d, a\rangle)$, $(a,\langle d,a,c,b\rangle)$ and $(a,\langle c,b,a,d\rangle)$. 
A memoryless strategy in $\widehat G$ (that is, a strategy that only records the order of last visits) must map each of these nodes to a unique successor node. Thus, such a strategy results in one of the paths
$(abc)^\omega$, $abc(ad)^\omega$, $(abcad)^\omega$, $(ad)^\omega$, $ad(abc)^\omega$, or $(adabc)^\omega$.

For $\gamma=0.26$ each of these paths has reward strictly smaller than $(abcabcad)^\omega$ and is hence not optimal.This concludes the proof.
\end{proof}

\end{document}